\documentclass[12pt,reqno]{amsart}
\usepackage[usenames,dvipsnames]{xcolor}
\usepackage{amsmath,amsfonts,amsthm,amssymb,color,tikz}
\usepackage[usenames,dvipsnames]{xcolor}
\usepackage{color}
\usepackage[T1]{fontenc}
\usepackage{enumerate}
\usepackage{hyperref}
\hypersetup{
     colorlinks   = true,
     citecolor    = blue,
     linkcolor    = blue
}
\usepackage{pdfsync}
\usepackage[font={scriptsize}]{caption}
\usepackage{comment}

\usepackage{subcaption, graphicx, framed}

\newcommand{\myfig}[4][ht]{
\begin{figure}[#1]
\centering
\includegraphics[#2]{#3}
\caption{#4\label{#3}}
\end{figure}
}

\usepackage[left=1in, right=1in, top=1.1in,bottom=1.1in]{geometry}
\newcommand{\R}{\mathbb R}

\newcommand{\PP}{\mathbb P}

\newcommand{\be}{\mathbf{E}}

\newcommand{\cf}{\mathcal F}

\newcommand{\cl}{\mathcal L}

\newcommand{\al}{\alpha}

\newcommand{\lp}{\left(}
\newcommand{\rp}{\right)}
\newcommand{\lc}{\left[}
\newcommand{\rc}{\right]}

\newtheorem{theorem}{Theorem}[section]

\theoremstyle{remark}

\theoremstyle{remark}

\newcommand{\bean}{\begin{eqnarray*}}
\newcommand{\eean}{\end{eqnarray*}}
\newcommand{\ben}{\begin{enumerate}}
\newcommand{\een}{\end{enumerate}}
\newcommand{\beq}{\begin{equation}}
\newcommand{\eeq}{\end{equation}}

\begin{document}

\title[Relativistic stable processes in heat conduction]{Relativistic stable processes in quasi-ballistic heat conduction in thin film semiconductors}

\address{Prakash Chakraborty: Department of Statistics,
  Purdue University,
  150 N. University Street,
  W. Lafayette, IN 47907,
  USA.}
\email{chakra15@purdue.edu}

\address{Bjorn Vermeersch:  CEA - Grenoble, France.
{Current address:} imec, Kapeldreef 75, 3001 Leuven, Belgium.}
\email{bjorn.vermeersch@imec.be}

\address{Ali Shakouri:  Birck Nanotechnology Center, 
Purdue University, 
West Lafayette, IN 47907, 
USA.}
\email{shakouri@purdue.edu}

\address{Samy Tindel: Department of Mathematics,
  Purdue University,
  150 N. University Street,
  W. Lafayette, IN 47907,
  USA.}
\email{stindel@purdue.edu}
\thanks{S. Tindel is supported by the NSF grant  DMS-1613163}

\author[P. Chakraborty, B. Vermeersch, A. Shakouri, S. Tindel]
{Prakash Chakraborty, Bjorn Vermeersch, Ali Shakouri, Samy Tindel}

\begin{abstract}
In this article, we show how relativistic alpha stable processes can be used to explain quasi-ballistic heat conduction in semiconductors. This is a method that can fit experimental results of ultrafast laser heating in alloys. It also provides a connection to a rich literature on Feynman-Kac formalism and random processes that transition from a stable L\'evy process on short time and length scales to the Brownian motion at larger scales. This transition was captured by a heuristic truncated L\'evy distribution in earlier papers. The rigorous Feynman-Kac approach is used to derive sharp bounds for the transition kernel. Future directions are briefly discussed.
\end{abstract}

\maketitle

\section{Introduction}
Standard heat flow, at a macroscopic level, is modeled by the random erratic movements of Brownian motions starting at the source of heat. However, this diffusive nature of heat flow predicted by Brownian motion is not seen in certain materials (semiconductors, dielectric solids) over short length and time scales \cite{1}. 
Experimental data portraying the non-diffusive behavior of heat flow has been observed for transient thermal grating (TTG) \cite{4,3}, time domain thermoreflectance (TDTR) \cite{7} and others \cite{6,5}, by altering the physical size of the heat source. The thermal transport in such materials is more akin to a superdiffusive heat flow, and necessitates the need for processes beyond Brownian motion to capture this heavy tail phenomenon.  Recent works \cite{15, 18, 17, 16, 20, 14, 19} try to explain the physics behind the quasiballistic heat dynamics. But these methods, driven mostly by the Boltzmann transport equation, are infeasible for processing experimental data. Some more recent studies \cite{26, 27} try to explain the non-diffusive heat flow through hyperbolic diffusion equations, however, closer investigation shows that these methods fail to capture the inherent onset of nondiffusive dynamics at short length scales in periodic heating regimes. 

The attempts mentioned above fail to provide a stochastic process that would explain the heat dynamics under short length-time regimes. The most natural stochastic process to explain a superdiffusive behavior is an alpha-stable L\'evy process \cite{Applebaum}. Alpha-stable L\'evy processes differ from Brownian motion in that its movements are governed by stable distributions as compared to Gaussian distributions for the latter. In this context, some of us have tried to explain the heat flow dynamics through a "truncated L\'evy distribution" approach \cite{PRBlevy1,PRBlevy2}, where it has been possible to extract the value of the L\'evy superdiffusive coefficient $\alpha$ that regulates the alloy's quasiballistic heat dynamics. 

The current contribution can be seen as a  further step in this direction. Specifically, let $T(t,x)$ designate the temperature of a semiconductor or dielectric solid in the experimental settings alluded to above, with initial condition $T_{0}(x)$. Then we shall describe $T$ through the following Feynman-Kac formula (see Section~\ref{sec:chr-FK} for more details about Feynman-Kac representations):
\beq\label{eq:FK}
T(t,x)
=
\be\lc T_{0}(x+X_{t})  \rc,
\eeq
where $X$ is a well-defined L\'evy process that captures the observed quasiballistic heat dynamics, in addition to being a good candidate for explaining the usual diffusive nature under non-special large length-time regimes. We shall see that such a process $X$ can be chosen as a so-called  relativistic stable process (see \cite{ryznar}, and \cite{carmona} for properties related to the relativistic Schr\"odinger operator). It possesses the remarkable property of behaving like an alpha-stable process under short length-time scales while being closer to Brownian motion otherwise. This is reflected in the estimates of the transition density, provided below in Section~\ref{sec:estimates}. Summarizing, our result lays the mathematical foundations of heat flow modeling on short time scales by means of stochastic processes. In addition, in spite of the fact that our computations are mostly one-dimensional, the model we propose allows natural generalizations to multidimensional and multilayer settings.

\section{Relativistic stable process: a primer}\label{sec:relativistic-primer}
In this section we give a short introduction on relativistic stable processes. We first recall the definition of this family of processes. Then we will give some kernel bounds indicating how relativistic processes transition, as $t$ increases, from an $\alpha$-stable behavior to a Brownian type behavior (this property being crucial to model quasiballistic heat dynamics in semiconductors).

\subsection{Characteristics and Feynman-Kac formula}\label{sec:chr-FK}
For any $M>0$ and $\al \in (0,2]$, a relativistic $\alpha$-stable process $X^M$ on $\R^d$ with a positive weight $M$ is a L\'evy process whose characteristic function is given, for any $t \geq 0$ and $\xi \in \R^d$, by:
\begin{equation}\label{eq:phi-m}
\phi_{M}(\xi)
\equiv
\mathbb{E}\lc \exp\lp \imath \xi \cdot (X_t^M - X_0^M)\rp\rc 
= \exp\lp -t\lp \lp |\xi|^2 + M^{2/\alpha}\rp^{\alpha/2}-M\rp\rp.
\end{equation}
When $M=0$, $X^M$ is simply a rotationally symmetric $\al$-stable process in $\R^d$. When $\al=2$, regardless of the value of $M$, the process is a Brownian motion. The infinitesimal generator of $X^M$ is {$\cl^{M} = M-(-\Delta + M^{2/\al})^{\al/2}$}. When $\al = 1$, this reduces to the free relativistic Hamiltonian $M-\sqrt{-\Delta + M^{2}}$, which explains the name of the process. An explicit expression for the L\'evy measure of $X^{M}$ can be found in \cite{Ch-12, AOP-Ch-12, Ch-11}. We omit this formula for sake of conciseness, since it will not be used in the remainder of the paper.

{ L\'evy processes like $X^M$ are classically used in order to represent solutions of deterministic PDEs. In our case, consider the following equation governing the temperature $T$ in our material:
\beq\label{eq:PDE}
\partial_t T(t,x) = \cl^M T(t,x), \text{ with } T(0,x) = T_0(x).
\eeq
Then it is a well known fact (see \cite{Applebaum}) that the solution $T$ to \eqref{eq:PDE} can be represented by the Feynman-Kac formula \eqref{eq:FK}, where the process $X^M$ is our relativistic $\alpha$-stable process. The Feynman-Kac representation is crucial in order to get equation~\eqref{eq:SPR} below.
}

\subsection{Transition kernel estimates}\label{sec:estimates}
In this subsection, we identify the behavior of a relativistic stable process with a stable process on short time scales and a Brownian motion on larger time scales. As mentioned above, this will be achieved by observing the patterns exhibited by the transition kernel of $X_{t}^{m}$. Some results will be stated without formal proof, and interested readers are referred to \cite{Ch-12, AOP-Ch-12, Ch-11} for more details.

Since $X^{M}$ is a Levy process, it is also a Markov process. As such it admits a transition kernel $p_{t}^{M}$, defined by:
\begin{equation*}
\PP\lp X_{s+t}^{M} \in A \, | \, X_{s} = x  \rp 
=
\int_{A} p_{{t}}^{M}(x,y) \, dy,
\end{equation*}
for all $x\in\R^{d}$ and $A\subset \R^{d}$. Notice that $p^{M}$ is related to the function $\phi_{M}$ (see definition \eqref{eq:phi-m}) as follows:
\begin{equation}\label{eq:kernel-as-inverse-fourier}
p_{{t}}^{M}(x,y)
=
\cf^{-1} \phi_{M}(x-y)
=
\frac{1}{(2 \pi)^d}
\int_{\R^d} e^{\imath (x-y)\cdot \xi} e^{-t\lbrace (M^{2/\alpha} + |\xi|^2)^{\alpha/2}-M \rbrace} d\xi.
\end{equation}

We start by a simple bound on $p^M_t$, exhibiting the stable behavior for small times and the Brownian behavior for large times. Observe that this bound does not depend on the space variables $x,y$. We include its proof, which is based on elementary considerations involving Fourier transforms, for sake of completeness.

\begin{theorem}\label{thm:transition-stable-to-brownian}
Consider a relativistic $\alpha$-stable process  $X^M$, and let $p^{M}$ be its transition kernel. 
Then there exists $c_1=c_1(\alpha)>0$ such that for all $t>0$ and all $x,y\in\R^{d}$:
\begin{equation}\label{ub_1}
p^M_t(x,y) \leq c_1(M^{d/\alpha-d/2}t^{-d/2} + t^{-d/\alpha}) .
\end{equation}
\end{theorem}

\begin{proof}
The strategy of our proof is based on the fact that the characteristic function $\phi_{M}$ defined by \eqref{eq:phi-m} behaves like a Gaussian characteristic function for low frequencies, and like an $\al$-stable characteristic function for high frequencies. We shall quantify this statement below.

\noindent
\textit{Step 1: Elementary inequalities:}
Let $\beta=\frac{\al}{2}$. The following inequality, valid for for $0 \leq z \leq 1$, is readily checked:
\begin{equation}\label{a1}
z^{\beta} - \beta z \leq 1 - \beta. 
\end{equation}
Substituting $z=\frac{M^{2/\alpha}}{|\xi|^2 + M^{2/\alpha}}$ in \eqref{a1}, we thus have:
\begin{equation*}
\dfrac{M}{(|\xi|^2 + M^{2/\alpha})^{\alpha/2}} - \dfrac{\alpha \, M^{2/\alpha}}{2(|\xi|^2 + M^{2/\alpha})} 
\leq 
1 - \frac{\alpha}{2} ,
\end{equation*}
which yields:
\begin{equation}\label{a2}
\lp |\xi|^2 + M^{2/\alpha}\rp^{\alpha/2} - M 
\geq 
\dfrac{\alpha |\xi|^2}{2(|\xi|^2 + M^{2/\alpha})^{1 - \alpha/2}}.
\end{equation}
Relation \eqref{a2} prompts us to split the frequency domain in two sets:
\beq\label{eq:dom-split}
A_1 = \lbrace \xi : |\xi|^2  \leq M^{2/\alpha}  \rbrace,
\quad\text{and}\quad
A_2 = \lbrace \xi :  |\xi|^2 > M^{2/\alpha}  \rbrace.
\eeq
Accordingly, we get the following lower bounds:
\begin{equation}\label{eq:a2}
(|\xi|^2 + M^{2/\alpha})^{\alpha/2} - M 
\geq
\begin{cases}
\dfrac{\alpha}{2(2M^{2/\alpha})^{1-\alpha/2}} \, |\xi|^2,  &{\text{when } \xi \in A_{1}} \\
\dfrac{\alpha}{2^{2-\alpha/2}} \, |\xi|^{\alpha}, &{\text{when } \xi \in A_{2}}.
\end{cases}
\end{equation}
This relation summarizes the separation between an $\al$-stable and a Gaussian regime alluded to above.

\noindent
\textit{Step 2: Consequence for the transition kernel.}
Recall relation \eqref{eq:kernel-as-inverse-fourier} for $p^{M}$, that is:
\begin{equation*}
p_{s}^{M}(x,y)
=
\frac{1}{(2 \pi)^d}
\int_{\R^d} e^{\imath (x-y)\cdot \xi} e^{-t\lbrace (M^{2/\alpha} + |\xi|^2)^{\alpha/2}-M \rbrace} d\xi.
\end{equation*}
In the integral above, we simply bound $|e^{\imath (x-y)\cdot \xi}|$ by 1 and split the integration domain $\R^{d}$ into $A_{1}\cup A_{2}$. Taking our relation \eqref{a2} into account, this yields:
\begin{equation}\label{a3}
p^M_t(x,y) 
 \leq \frac{1}{(2 \pi)^d} \int_{A_1} e^{-\frac{\alpha  t}{2(2M^{2/\alpha})^{1-\alpha/2}} \, |\xi|^2} d\xi \\
 + \frac{1}{(2 \pi)^d} \int_{A_2} e^{-\frac{\alpha t}{2^{2-\alpha/2}} \, |\xi|^{\alpha}} d\xi
 \leq  
I_{t}^{1} + I_{t}^{2},
\end{equation}
where 
\begin{eqnarray*}
I_{t}^{1} & = &
\frac{1}{(2 \pi)^d} \left[ \frac{1}{t^{d/2}} \left(\frac{(2M^{2/\alpha -1})}{\alpha/2}\right)^{d/2}
\int_{\R^d} e^{-|\xi|^2} d\xi \right] \\
I_{t}^{2}  &=& \frac{1}{(2 \pi)^d} \left[ \frac{1}{t^{d/\alpha}} \lp\frac{(2^{2-\alpha/2})}{\alpha}\rp^{d/\alpha}\int_{\R^d} e^{-|\xi|^{\alpha}} d\xi \right] .
\end{eqnarray*}
It is now easily checked that
\begin{equation*}
I_{t}^{1} = \frac{c_{2}}{t^{d/2}},
\quad\text{and}\quad
I_{t}^{2} = \frac{c_{3}}{t^{d/\alpha}}.
\end{equation*}
Plugging this information into \eqref{a3}, our claim \eqref{ub_1} follows.
\end{proof}

The upper bound \eqref{ub_1} already captures a lot of the information we need on relativistic stable processes. Invoking sophisticated arguments based on stopping times and Dirichlet forms, one can get upper and lower bounds on the transition kernel $p^{M}$ involving some exponential decay in the space variables $x,y$. We summarize those refinements in the following theorem.

\begin{theorem}
Let $p^{M}$ be the transition kernel defined by \eqref{eq:kernel-as-inverse-fourier}. Then the following estimates hold true.

\noindent
\emph{(i) Small time estimates.}
Let $T>0$ be a fixed time horizon. Then there exists $C_1>0$ such that  for all $t \in (0,T]$ and $x,y \in \R^d$,
\begin{equation}\label{eq:small-time-bnd}
C_{1}^{-1}\lp t^{-d/\alpha} \wedge  \dfrac{t e^{-M^{1/\al}|x-y|}}{|x-y|^{{(d+\al+1)}/{2}}} \rp \leq p^M_t(x,y) \leq C_{1}\lp t^{-d/\alpha} \wedge  \dfrac{t e^{-M^{1/\al}|x-y|}}{|x-y|^{{(d+\al+1)}/{2}}} \rp .
\end{equation}

\noindent
\emph{(ii) Large time estimates.}
There exists $C_2 \geq 1$ such that for every  $t \in [1, \infty)$ and $x,y \in \R^d$,
\begin{equation}
C_2^{-1} t^{-d/2} \leq p^M_t(x,y) \leq C_2 t^{-d/2} .
\end{equation}
\end{theorem}

\section{Application of relativistic stable processes to thermal modelling}

In this section we show how to apply the mathematical formalism of Section \ref{sec:relativistic-primer} to our concrete physical setting. More specifically, in Section \ref{sec:thermal-prop} we shall introduce length scales in our L\'evy exponent \eqref{eq:phi-m}. Then Section \ref{sec:tdtr-measurement} is devoted to a description of our experimental setting, and also relates our measurements to the Fourier exponents we have put forward. 

\subsection{Formulation in terms of material thermal properties}\label{sec:thermal-prop}
The aforementioned evolution of relativistic processes, from alpha-stable behaviour at short length and time scales to regular Brownian motion at longer scales, renders them suitable to describing quasiballistic thermal transport in semiconductor alloys. Let us consider such a material, having nominal thermal diffusivity $D = \kappa/C_v$ with $\kappa$ being the thermal conductivity and $C_v$ the volumetric heat capacity (in { ${\rm Jm}^{-3}{\rm K}^{-1}$}). 
{ The physical quantity we have access to is a slight variation of the function $T(t,x)$ defined by \eqref{eq:FK}. Specifically the single pulse response for the d-dimensional excess thermal energy can be expressed as $P(t,x) = C_v \, \Delta T(t,x)$. Under the L\'evy flight paradigm the Fourier transform of $P$ is written as
\beq\label{eq:SPR}
P(| \xi | , t) = \exp \lp- t \, \psi(| \xi |) \rp,
\eeq
for a given L\'evy exponent (also called special heat propagator) $\psi$. For the relativistic case under study here, this spatial heat propagator is simply a multiple of the function $\phi_M$ introduced in \eqref{eq:phi-m}. Namely $\psi$ reads}
\beq\label{eq:cf-D_al}
\psi(| \xi |) = D_{\alpha} \, \left[ \left(| \xi |^2 + M^{2/\alpha} \right)^{\alpha/2} - M \right]. 
\eeq
The prefactor $D_{\alpha}$ with unit m$^{\alpha}$/s denormalises the characteristic function for dimensionless space and time variables defined by Eq. (\ref{eq:phi-m}) to its physical counterpart, and denotes the fractional diffusivity of the alpha-stable regime as we shall see shortly.
\par
For thermal modelling purposes it is furthermore convenient to reformulate the process mass $M$, which has an exponent-dependent unit 1/m$^{\alpha}$, in terms of an associated characteristic length scale $x_0$ around which the transition from alpha-stable (L\'evy) to Brownian dynamics takes place{ . Notice that according to our previous analysis leading up to Eq. \eqref{eq:dom-split} and \eqref{eq:a2}} this length should be given by:
\begin{equation}
x_0 = | \xi_0 |^{-1} = M^{-1/\alpha}.
\end{equation}
This means that expression \eqref{eq:cf-D_al} can be recast as
\begin{align}\label{eq:mod-M}
\psi(| \xi |) &= D_{\alpha} \, \left[ \left(| \xi |^2 + | \xi_0 |^2 \right)^{\alpha/2} - | \xi_0 |^{\alpha} \right] \nonumber \\
&= D_{\alpha} \, | \xi_0 |^{\alpha} \, \left[ \left( \tilde{\xi}^2 + 1 \right)^{\alpha/2} - 1 \right],  
\end{align}
where $\tilde{\xi} \equiv | \xi | / | \xi_0|$.
With those values of $M$ and $\xi_0$ in hand, we can translate \eqref{eq:a2} into an asymptotic transport limit as follows:
\begin{align}
\text{alpha-stable regime } \tilde{\xi} \gg 1&:~\psi(| \xi |) \simeq D_{\alpha} \, | \xi |^{\alpha}, \nonumber \\
\text{Brownian regime } \tilde{\xi} \ll 1&:~\psi(| \xi |) \simeq \frac{\alpha D_{\alpha}}{2 | \xi_0 |^{2-\alpha}} \, | \xi |^2. \label{eq:brown-regime}
\end{align}
The former corresponds to L\'evy superdiffusion with characteristic exponent $\alpha$ and fractional diffusivity $D_{\alpha}$; the latter should recover to nominal diffusive transport { $\psi (|\xi|) \equiv D | \xi |^2$. In order to make the last relation compatible with \eqref{eq:brown-regime} we must set}
\beq\label{eq:D-al}
\frac{\alpha D_{\alpha}}{2 | \xi_0 |^{2-\alpha}} = D \quad \Longrightarrow \quad D_{\alpha} = \frac{2D}{\alpha x_0^{2-\alpha}}.
\eeq
Finally, { plugging \eqref{eq:D-al} into \eqref{eq:mod-M}} the heat propagator reads
\begin{equation}
\psi(| \xi |) = \frac{2D}{\alpha x_0^2} \, \left[ \left(1 + x_0^2 | \xi |^2 \right)^{\alpha/2} - 1 \right]. \label{propagator}
\end{equation}
This formulation contains 3 material dependent parameters, each with an intuitive physical meaning: the characteristic exponent $\alpha$ of the alpha-stable regime; the nominal diffusivity $D$ of the Brownian regime; and the characteristic length scale $x_0$ around which the transition between those two asymptotic limits occurs (Fig. \ref{figure_propagator}). In the sections that follow, we determine these parameter values for In$_{0.53}$Ga$_{0.47}$As by fitting a thermal model built upon the propagator (\ref{propagator}) to time-domain thermoreflectance (TDTR) measurement signals.
\myfig[!htb]{width=0.7\linewidth}{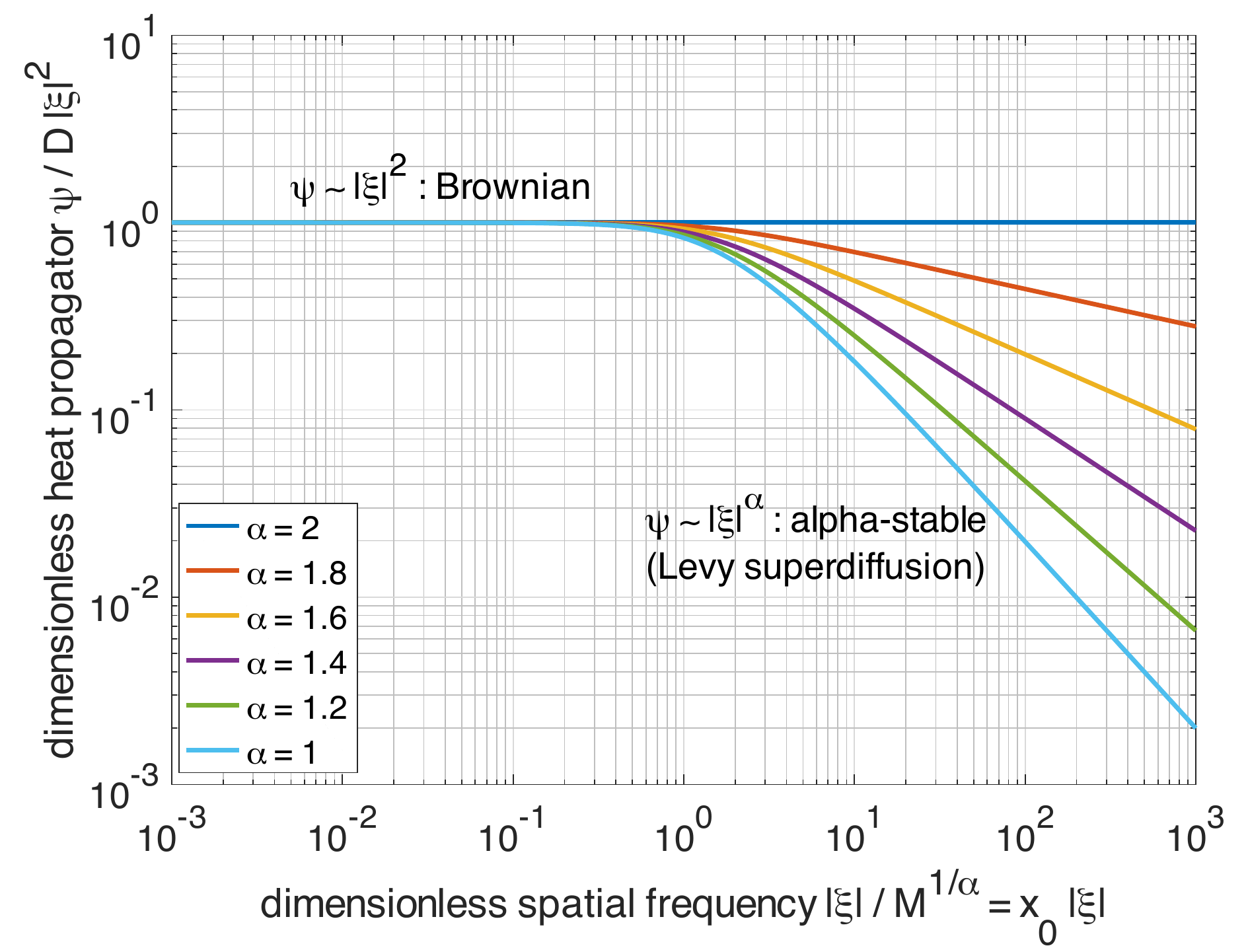}{Transition between Brownian and alpha-stable L\'evy behavior.} 

\subsection{Modelling of TDTR measurement signals}\label{sec:tdtr-measurement}
The central principle in TDTR is to heat up the sample with ultrashort \textit{pump} laser pulses, and then monitor the thermal transient decay using a \textit{probe} beam. Pulses from the laser are split into a pump beam and probe beam. The pump pulses pass through an electro-optic modulator (EOM) before being focused onto the sample surface through a microscope objective. A thin (50-100 nm) aluminium film is deposited onto the sample to act as measurement transducer: the metal efficiently absorbs the pump light and converts it to heat, and translates temperature variations to changes in surface reflectivity which can be captured by the probe. Lock-in detection at the pump modulation frequency $f_{\text{mod}}$ resolves the thermally induced reflectivity changes captured by the probe beam. A mechanical delay stage allows to vary the relative arrival time of the pump and probe pulses at the sample with picosecond resolution. To minimise the impact of random fluctuations in laser power and the variation of the pump beam induced by the delay stage, thermal characterisation is performed not on the raw lock-in signal itself but rather the ratio $-V_{\text{in}}/V_{\text{out}}$ of the in-phase and out-of-phase components as a function of the pump-probe delay.
\par
Theoretical ratio curves $-V_{\text{in}}/V_{\text{out}}$ can be computed semi-analytically through mathematical manipulation of the semiconductor single-pulse response (\eqref{eq:SPR}), as described in detail in Refs. \cite{cahillRSI,schmidt,JAPmultiD}. Briefly, we first obtain the surface temperature response of a semi-infinite semiconductor to a cylindrically symmetric energy input via Fourier inversion of (\eqref{eq:SPR}) with respect to the cross-plane coordinate. Next, a matrix formalism that accounts for heat flow in the metal transducer and across the intrinsic thermal resistivity $r_{\text{ms}}$ (in K-m$^2$/W) of the metal-semiconductor interface provides the temperature response, weighted by the Gaussian probe beam, of the transducer top surface induced by a Gaussian pump pulse. Finally, harmonic assembly of this response at frequencies $n \cdot f_{\text{rep}} \pm f_{\text{mod}} \quad (\text{for } n = 0,1,\ldots)$ accounting for the laser repetition rate $f_{\text{rep}}$, pump modulation frequency $f_{\text{mod}}$, and phase factors induced by the pump-probe delay $\tau$ yields the theoretical lock-in ratio signal $-V_{\text{in}}/V_{\text{out}}(\tau)$.
\section{Experimental analysis}
We have applied our model to TDTR measurements taken on a $\simeq$2 micron thick film of In$_{0.53}$Ga$_{0.47}$As ($C_v \simeq 1.55\,$MJ/m$^3$-K) that was MBE-grown on a lattice-matched InP substrate. We note that although the semiconductor alloy under study (the InGaAs layer) is a geometrically thin film, thermally speaking it can still be considered, as is assumed by the thermal model, as a semi-infinite layer with good approximation. This is because the effective thermal penetration length $\ell = \sqrt{D/(\pi f_{mod})}$ stays firmly within the film over the experimentally probed modulation range $0.8 \, \text{MHz} \lesssim f_{mod} \lesssim 18\, \text{MHz}$. The aluminium transducer deposited onto the sample measured 64$\,$nm in thickness as determined by picoseconds accoustics. We used pump and probe beams with $1/e^2$ radii at the focal plane of 6.5 and 9 microns respectively, with respective powers of 17 and 8 mW at the sample surface.
\par
In the thermal model with relativistic stable heat propagator (\ref{propagator}), we fixed the heat capacity at the aforementioned 1.55$\,$MJ/m$^3$-K. Theoretical ratio curves were then collectively fitted through nonlinear least-square optimisation to signals measured at 7 different modulation frequencies to identify the 4 key thermal parameters: the characteristic exponent $\alpha$ of the L\'evy superdiffusion regime; the quasiballistic-diffusive transition length scale $x_0$ associated to the mass $M$; the nominal thermal conductivity $\kappa = C_v D$ of the diffusive regime; and the thermal resistivity $r_{ms}$ of the transducer/semiconductor interface. The resulting best-fitting values $\alpha$ = 1.695 , $x_0$ = 0.86$\,\mu$m , $\kappa$ = 5.82$\,$W/m-K , $r_{ms}$ = 4.28$\,$nK-m$^2$/W yield an excellent agreement with the measured signals (Fig. \ref{ratiocurves_bestfit}). Theoretical curves with parameter values deviating from the best fitting ones (Fig. \ref{ratiocurves_tolerance}) furthermore visually reveal the sensitivity to each of the parameters and illustrate the good quality of the best fit.
\myfig[!htb]{width=0.7\linewidth}{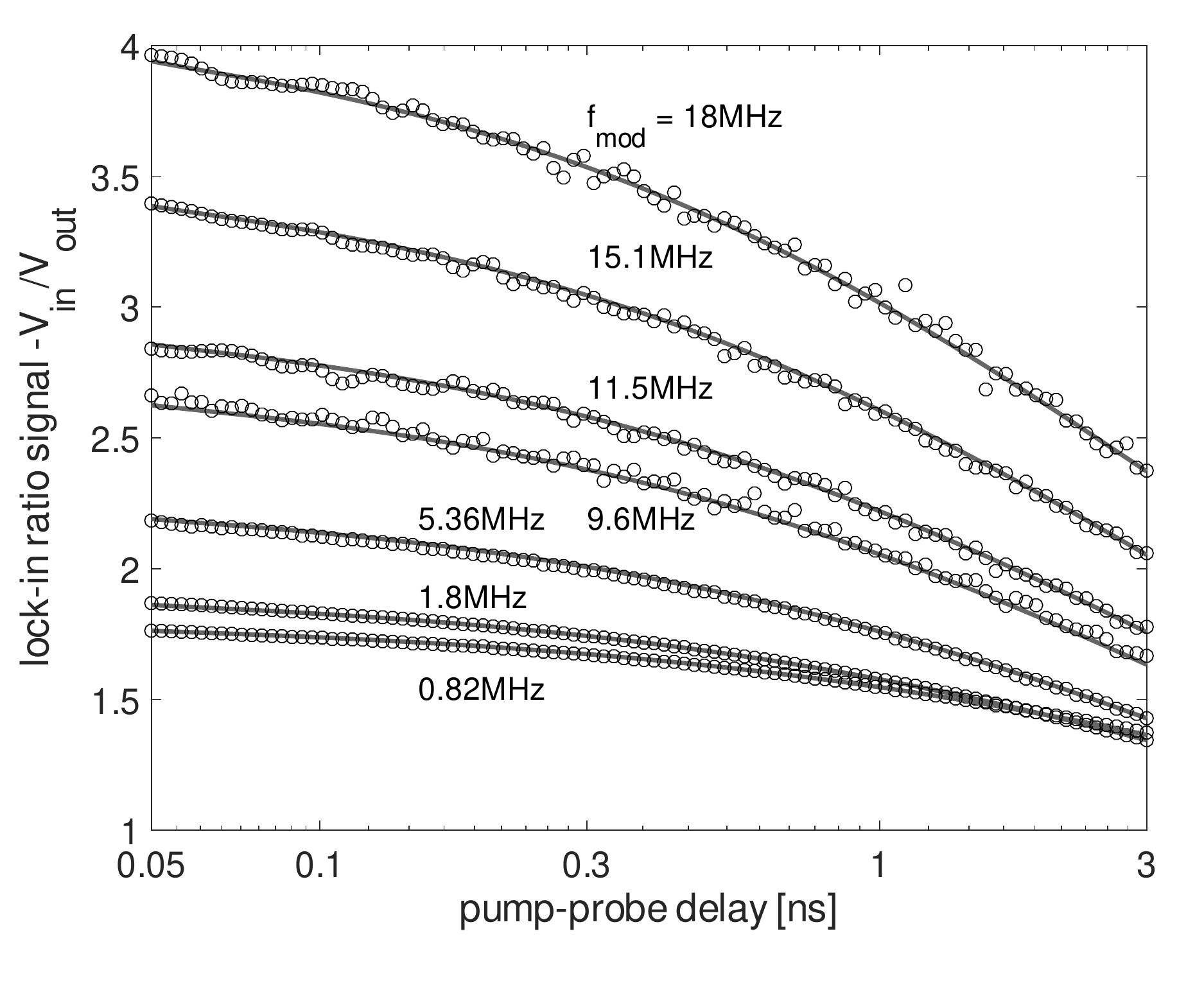}{TDTR characterisation of Al/InGaAs sample: a thermal model based on a relativistic stable random motion of heat (lines) provides an excellent fit to measured signals (symbols).} 
\begin{figure*}[t]
\centering
\begin{subfigure}[b]{0.5\linewidth}
\includegraphics[width=\linewidth]{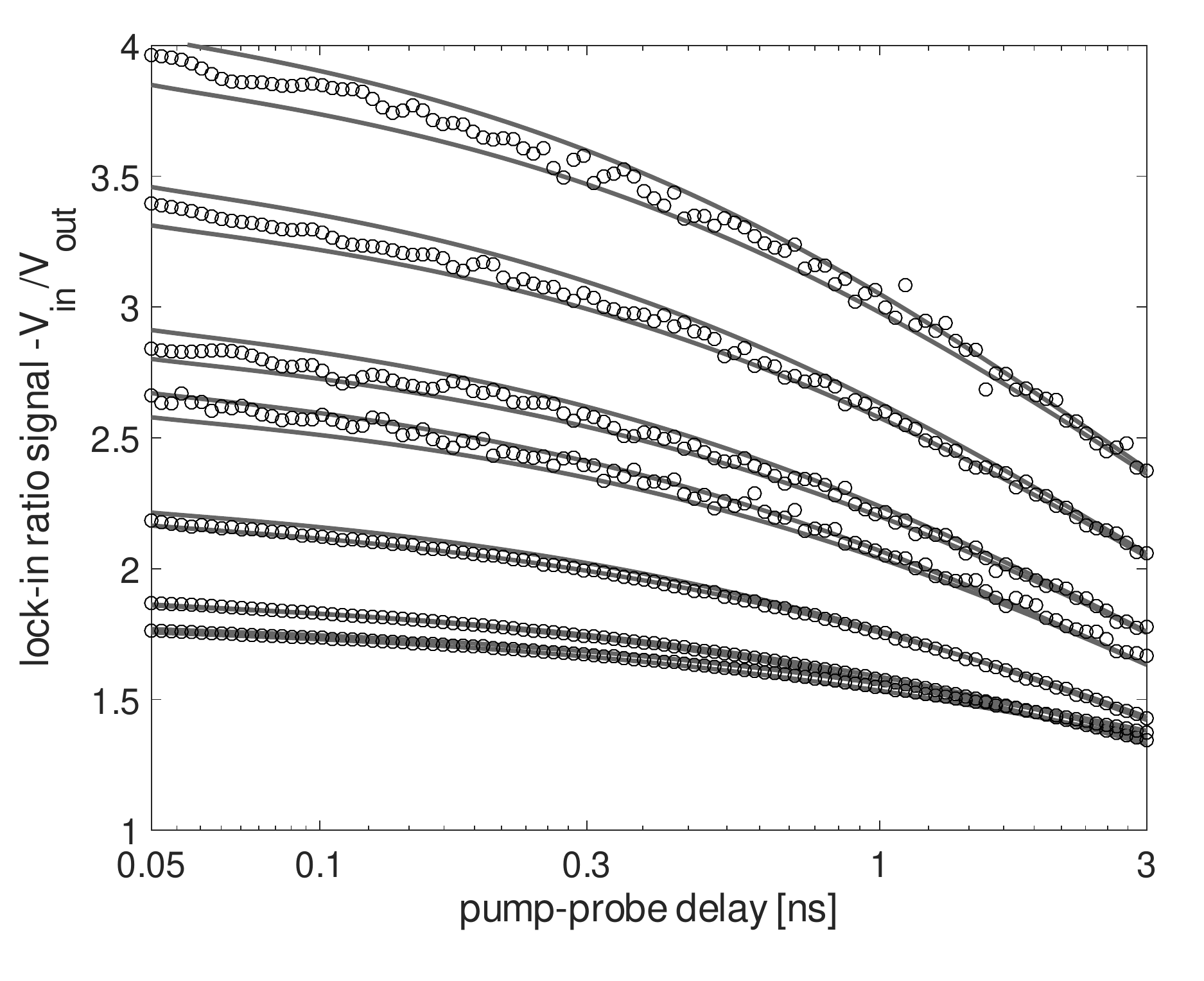}
\caption{$\alpha \pm 0.02$}
\label{tolerance-alpha}
\end{subfigure}%
\begin{subfigure}[b]{0.5\linewidth}
\includegraphics[width=\linewidth]{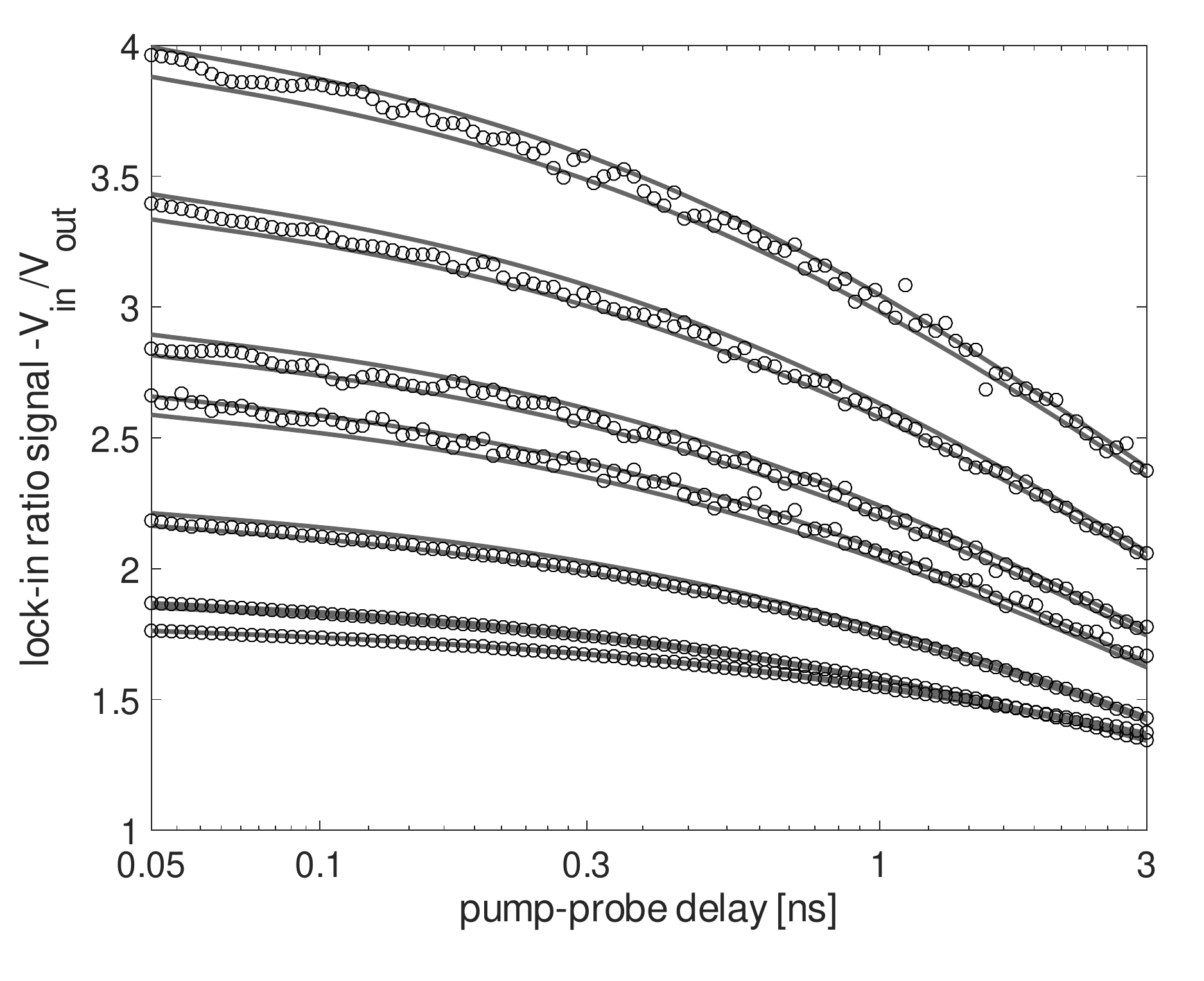}
\caption{$x_0 \pm 10 \%$}
\label{tolerance-x0}
\end{subfigure}%

\begin{subfigure}{0.5\linewidth}
\includegraphics[width=\linewidth]{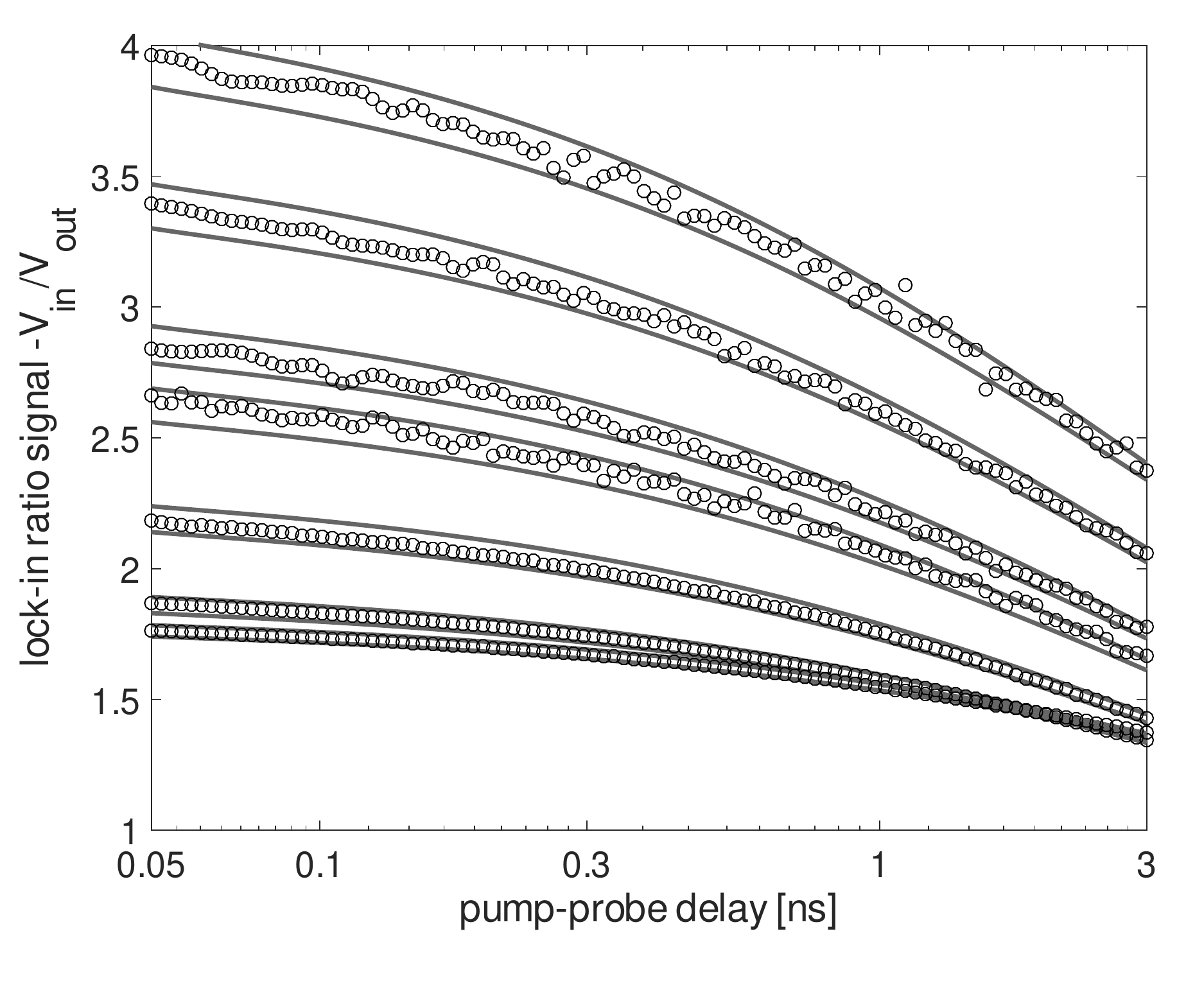}
\caption{$\kappa \pm 5\%$}
\label{tolerance-kappa}
\end{subfigure}%
\begin{subfigure}{0.5\linewidth}
\includegraphics[width=\linewidth]{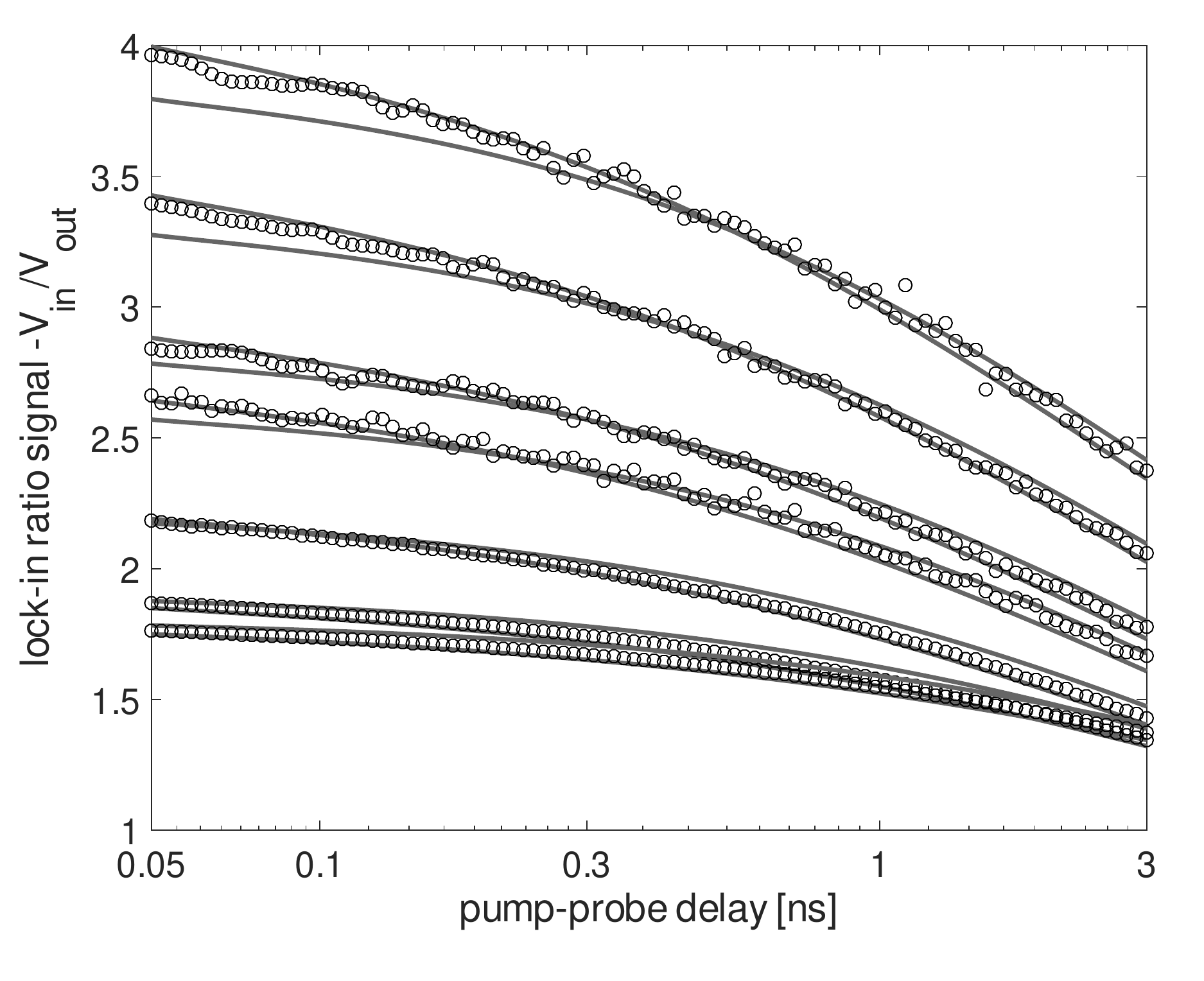}
\caption{$r_{ms} \pm 50 \%$}
\label{tolerance-rms}
\end{subfigure}%
\caption{Ratio curve fitting tolerance and sensitivity. The plotted theoretical curves are computed with sub-optimal parameter combinations in which one parameter deviates from its best fitting value as indicated.}
\label{ratiocurves_tolerance}
\end{figure*}
\section{Conclusions and Outlook}
Quasi-ballistic heat propagation in materials can be studied using atomic parameters through a multitude of techniques. First principle calculations and multi spectral phonon Boltzmann transport equations are very powerful in this regard. However, their use in the study of heat propagation in multi-layer/anisotropic materials and materials with complex geometries is limited. The Feynman-Kac representation of solutions to partial differential equations with non-local parameters can potentially provide alternative approaches to explain experimental thermal data. In this article we have replaced the traditional heat equation by a different PDE, whose solution has a Feynman-Kac representation driven by the so-called relativistic stable L\'evy process. The transition characteristics of this process is in harmony to the heat propagation behaviour exhibited by TDTR data. In general, numerical approximations of the PDE solution can also be achieved through Monte Carlo simulations of the driving stochastic process in the Feynman-Kac formula. In particular, these numerical computations may provide substitute techniques to optimize materials or source geometry in order to reduce heating from nanoscale and/or ultrafast devices.

Our next challenge in this direction will be to model multidimensional transport in multilayer structures. To this aim, we shall investigate two methods: (i) Monte Carlo simulation according to our Feynman-Kac representation \eqref{eq:FK}, taking into account jumps and change of media. (ii) Related PDEs involving the non local operator $\cl^{M} = M-(-\Delta + M^{2/\al})^{\al/2}$, with boundary terms corresponding to the different layers. Both methods rely crucially on the relativistic L\'evy representation advocated in this paper. They will be subject of future publications. 

\newpage

\bibliographystyle{plain}
\bibliography{references}

\end{document}